\documentclass[12pt, reqno]{amsart} 
\usepackage{amsmath, amsthm, amscd, amsfonts, amssymb, graphicx, color}
\usepackage[bookmarksnumbered, colorlinks, plainpages]{hyperref}

\newtheorem{theorem}{Theorem}[section]
\newtheorem{lemma}[theorem]{Lemma}
\newtheorem{proposition}[theorem]{Proposition}
\newtheorem{corollary}[theorem]{Corollary}
\theoremstyle{definition}
\newtheorem{definition}[theorem]{Definition}
\newtheorem{example}[theorem]{Example}

\theoremstyle{remark}

\numberwithin{equation}{section}

\begin{document}
\title[$(1-2u^2)$-constacyclic codes over $\mathbb{F}_p+u\mathbb{F}_p+u^2\mathbb{F}_p$]{$(1-2u^2)$-constacyclic codes over $\mathbb{F}_p+u\mathbb{F}_p+u^2\mathbb{F}_p$}
\author[Mostafanasab and Karimi]{Hojjat Mostafanasab and Negin Karimi}


\begin{abstract}
Let $\mathbb{F}_p$ be a finite field and $u$ be an indeterminate. This article studies $(1-2u^2)$-constacyclic codes over the ring $\mathbb{F}_p+u\mathbb{F}_p+u^2\mathbb{F}_p$, where $u^3=u$.  We describe generator polynomials of 
this kind of codes and investigate the structural properties of these codes by a decomposition theorem.
\end{abstract}

\maketitle

\section{Introduction}
Error-Correcting codes play important roles in applications ranging from data networking to satellite 
communication to compact disks. Most coding theory concerns
on linear codes since they have clear structure that makes them simpler to discover, to
understand and to encode and decode.
Codes over finite rings have been studied since the early 1970s. 
Recently codes over rings have
generated a lot of interest after a breakthrough paper by Hammons et al. \cite{HK} showed
that some well known binary non-linear codes are actually images of some linear
codes over $\mathbb{Z}_4$ under the Gray map.
Cyclic codes are amongst the most studied algebraic codes.
Their structure is well known over finite fields \cite{M}. 
Constacyclic codes over finite fields form a remarkable class of linear codes, as they include the important family of cyclic codes. Constacyclic codes also have practical applications as they can be efficiently encoded using simple shift registers. They have rich algebraic structures for efficient error detection and correction, which explains their preferred role in engineering.
In general, due to their rich algebraic structure, constacyclic codes have been studied over various
finite chain rings (see \cite{AN},\cite{BU1}-\cite{DL},\cite{QZZ}-\cite{ZW}). In \cite{ZW}, Zhu and Wang investigated 
$(1-2u)$-constacyclic codes over $\mathbb{F}_p+v\mathbb{F}_p$, where $v^2=v$.
In \cite{G,KYS,LSS}, some kind of codes over $\mathbb{F}_p+u\mathbb{F}_p+u^2\mathbb{F}_p$, where $u^3=u$, have been studied.
The present paper is devoted to a class of constacyclic codes over
$\mathbb{F}_p+u\mathbb{F}_p+u^2\mathbb{F}_p$, i.e., $(1-2u^2)$-constacyclic codes over $\mathbb{F}_p+u\mathbb{F}_p+u^2\mathbb{F}_p$.

Let $\sigma,\gamma$ and $\varrho$ be maps from $\mathcal{R}^n$ to $\mathcal{R}^n$ given by
$$\hspace{-1.1cm}\sigma(r_0,r_1,\dots,r_{n-1})=(r_{n-1},r_0,r_1,\dots,r_{n-2}),$$
$$\gamma(r_0,r_1,\dots,r_{n-1})=(-r_{n-1},r_0,r_1,\dots,r_{n-2}),~~~~~ \mbox{and}$$
$$\hspace{4mm}\varrho(r_0,r_1,\dots,r_{n-1})=((1-2u^2)r_{n-1},r_0,r_1,\dots,r_{n-2}),$$
respectively. 
Let $\mathcal{C}$ be a linear code of lenght $n$ over $\mathcal{R}$. Then $\mathcal{C}$ is said to be cyclic if $\sigma(\mathcal{C}) =\mathcal{C}$, negacyclic if $\gamma(\mathcal{C})=\mathcal{C}$ and $(1-2u^2)$-constacyclic if $\varrho(\mathcal{C})=\mathcal{C}$. 

Let $\mathcal{C}$ be a code of length $n$ over $\mathcal{R}$, and $P(\mathcal{C})$ be its polynomial
representation, i.e., $$P(\mathcal{C})=\Big\{\sum\limits_{i=0}^{n-1}r_ix^i|(r_0,\dots,r_{n-1})\in\mathcal{C}\Big\}.$$
It is easy to see that:
\begin{theorem}
A code $\mathcal{C}$ of length $n$ over $\mathcal{R}$ is $(1-2u^2)$-constacyclic if and only if $P(\mathcal{C})$ is an ideal of $\mathcal{R}[x]/\langle x^n-(1-2u^2)\rangle$.
\end{theorem}
Let $x=(x_0,x_1,\dots,x_{n-1})$ and $y=(y_0,y_
1,\dots,y_{n-1})$ be two elements of $\mathcal{R}^n$. The Euclidean inner product of $x$ and $y$ in $\mathcal{R}^n$ is defined as $x\cdot y=x_0y_
0+x_1y_1 +\dots+x_{n-1}y_{n-1}$, where the operation is performed in $\mathcal{R}$. The dual code of $\mathcal{C}$ is defined as
$C^\bot=\{x\in\mathcal{R}^n| x\cdot y=0 \mbox{ for every } y\in\mathcal{C}\}$.

We define the Gray map $\Phi:\mathcal{R}\to\mathbb{F}_p^2$ by $a+bu+cu^2\mapsto(-c,2a+c)$. This map can be extended to $\mathcal{R}^n$ in a natural way: 
$$\hspace{-11.3cm}\Phi:\mathcal{R}^n\to\mathbb{F}_p^{2n}$$
$$(r_0,r_1,\dots,r_{n-1})\mapsto(-c_0,-c_1,\dots,-c_{n-1},2a_0+c_0,2a_1+c_1,\dots,2a_{n-1}+c_{n-1})$$
where $r_i=a_i+b_iu+c_iu^2$, $0\leq i\leq n-1$.

We denote by $\eta_1,\eta_2,\eta_3$ respectively the following elements of $\mathcal{R}$:
$$\eta_1=1-u^2,~~~\eta_2=2^{-1}(u+u^2),~~~\eta_3=2^{-1}(-u+u^2).$$
Note that $\eta_1,\eta_2$ and $\eta_3$ are mutually orthogonal idempotents over $\mathcal{R}$ and $\eta_1+\eta_2+\eta_3=1$.
Let $\mathcal{C}$ be a linear code of length $n$ over $\mathcal{R}$. Define
\begin{eqnarray*}
\mathcal{C}_1&=&\{x\in\mathbb{F}_p^n\mid\exists y,z\in\mathbb{F}_p^n,\eta_1x+\eta_2y+\eta_3z\in\mathcal{C}\},\\
\mathcal{C}_2&=&\{y\in\mathbb{F}_p^n\mid\exists x,z\in\mathbb{F}_p^n,\eta_1x+\eta_2y+\eta_3z\in\mathcal{C}\},\\
\mathcal{C}_3&=&\{z\in\mathbb{F}_p^n\mid\exists x,y\in\mathbb{F}_p^n,\eta_1x+\eta_2y+\eta_3z\in\mathcal{C}\}.
\end{eqnarray*}
Then $\mathcal{C}_1,\mathcal{C}_2$ and $\mathcal{C}_3$ are all linear codes of length $n$ over $\mathbb{F}_p$. Moreover, the code $\mathcal{C}$ of length $n$
over $\mathcal{R}$ can be uniquely expressed as
$\mathcal{C}=\eta_1\mathcal{C}_1\oplus\eta_2\mathcal{C}_2\oplus\eta_3\mathcal{C}_3.$


\section{Main results}

\bigskip

\begin{theorem}\label{T2}
Let $\varrho$ denote the $(1-2u^2)$-constacyclic shift of $\mathcal{R}^n$ and $\sigma$ the cyclic shift of $\mathbb{F}^{2n}_p$. 
If $\Phi$ is the Gray map of $\mathcal{R}^n$ into $\mathbb{F}^{2n}_p$, then $\Phi\varrho=\sigma\Phi$.
\end{theorem}
\begin{proof}
Let $\bar{r}=(r_0,r_1,\dots,r_{n-1})\in\mathcal{R}^n$ where $r_i=a_i+b_iu+c_iu^2$ with $a_i,b_i,c_i\in\mathbb{F}_p$
for $0\leq i\leq n-1$. Taking $(1-2u^2)$-constacyclic shift on $\bar{r}$, we have 

$$\hspace{-4.6cm}\varrho(\bar{r})=\big((1-2u^2)r_{n-1},r_0,r_1,\dots,r_{n-2}\big)$$
$$=\big(a_{n-1}-b_{n-1}u+(-2a_{n-1}-c_{n-1})u^2,a_0+b_0u+c_0u^2,$$
$$a_1+b_1u+c_1u^2,\dots,a_{n-2}+b_{n-2}u+c_{n-2}u^2\big).$$

Now, using the definition of Gray map $\Phi$, we can deduce that

$$\hspace{0cm}\Phi(\varrho(\bar{r}))=\big(2a_{n-1}+c_{n-1},-c_0,-c_1,\dots,-c_{n-2},2a_{n-1}+(-2a_{n-1}-c_{n-1}),$$
$$\hspace{-6mm}2a_0+c_0,2a_1+c_1,\dots,2a_{n-2}+c_{n-2}\big).$$
On the other hand,
\begin{eqnarray*}
\sigma(\Phi(\bar{r}))&=&\sigma(-c_0,-c_1,\dots,-c_{n-1},2a_0+c_0,2a_1+c_1,\dots,2a_{n-1}+c_{n-1})\\
&=&\big(2a_{n-1}+c_{n-1},-c_0,-c_1,\dots,-c_{n-1},2a_0+c_0,2a_1+c_1,\dots,2a_{n-2}+c_{n-2}\big).
\end{eqnarray*}
Therefore,
$$\Phi\varrho=\sigma\Phi.$$
\end{proof}
\begin{theorem}\label{T3}
The Gray image of a $(1-2u^2)$-constacyclic code over $\mathcal{R}$ of lenght $n$ is a 
cyclic code over $\mathbb{F}_p$ of lenght $2n$.
\end{theorem}
\begin{proof}
Let $\mathcal{C}$ be a $(1-2u^2)$-constacyclic code over $\mathcal{R}$. Then $\varrho(\mathcal{C})=\mathcal{C}$, and therefore, $(\Phi\varrho)(\mathcal{C})=\Phi(\mathcal{C})$. It follows
from Theorem \ref{T2} that $\sigma(\Phi(\mathcal{C}))=\Phi(\mathcal{C})$, which means that $\Phi(\mathcal{C})$ is a cyclic code.
\end{proof}


Notice that $(1-2u^2)^n=1-2u^2$ if $n$ is odd and $(1-2u^2)^n=1$ if $n$ is even.
\begin{proposition}
Let $\mathcal{C}$ be a code of lenght $n$ over $\mathcal{R}$.
Then $\mathcal{C}$ is a $(1-2u^2)$-constacyclic code if and only if  $\mathcal{C}^{\bot}$  is a $(1-2u^2)$-constacyclic code.
\end{proposition}
\begin{proof}
The ``only if'' part follows from Proposition 2.4 of \cite{D3}. For the converse note the fact that $(\mathcal{C}^{\bot})^{\bot}=\mathcal{C}$.
\end{proof}

Recall that a code $\mathcal{C}$ is said to be self-orthogonal provided $\mathcal{C}\subseteq\mathcal{C}^\bot$.
\begin{proposition}
Let $\mathcal{C}$ be a code of length $n$ over $\mathcal{R}$ such that $\mathcal{C}\subset\big(\mathbb{F}_p+u^2\mathbb{F}_p\big)^n$. 
If $\mathcal{C}$ is self-orthogonal, then so is $\Phi(\mathcal{C})$.
\end{proposition}
\begin{proof}
Assume that $\mathcal{C}$ is self-orthogonal. Let $r_1=a_1+c_1u^2,~r_2=a_2+c_2u^2\in\mathcal{C}$, where $a_i,c_i\in\mathbb{F}_p^n$ for $i=1,2$.
Now by Euclidean inner product of $r_1$ and $r_2$, we have
\begin{eqnarray*}
 r_1\cdot r_2&=&(a_1+c_1u^2)\cdot(a_2+c_2u^2)\\
 &=&a_1a_2+(a_1c_2+c_1a_2+c_1c_2)u^2.
\end{eqnarray*}
If $r_1\cdot r_2=0$, then $a_1a_2=a_1c_2+c_1a_2+c_1c_2=0$. 
Therefore 
\begin{eqnarray*}
 \Phi(r_1)\cdot\Phi(r_2)&=&(-c_1,2a_1+c_1)\cdot(-c_2,2a_2+c_2)\\
 &=&4a_1a_2+2(c_1c_2+a_1c_2+c_1a_2)=0.
\end{eqnarray*}
Hence $\Phi(\mathcal{C}^{\bot})\subseteq\Phi(\mathcal{C})^{\bot}$. Consequently $\Phi(\mathcal{C})\subseteq\Phi(\mathcal{C})^{\bot}$.
\end{proof}

\begin{theorem}\label{T5}
Let $\mathcal{C}=\eta_1\mathcal{C}_1\oplus\eta_2\mathcal{C}_2\oplus\eta_3\mathcal{C}_3$ be a code of length $n$ over $\mathcal{R}$. 
Then $\mathcal{C}$ is a $(1-2u^2)$-constacyclic code of length $n$ over $\mathcal{R}$ if and only if $\mathcal{C}_1$ is cyclic and 
$\mathcal{C}_2$, $\mathcal{C}_3$ are negacyclic codes of length $n$ over $\mathbb{F}_p$.
\end{theorem}

\begin{proof}
First of all notice that $(1-2u^2)\eta_1=\eta_1$, $(1-2u^2)\eta_2=-\eta_2$ and $(1-2u^2)\eta_3=-\eta_3$. Let $\bar{r}=(r_0,r_1,\dots,r_{n-1})\in\mathcal{C}$.
Then $r_i=\eta_1a_i+\eta_2b_i+\eta_3c_i$, where $a_i,b_i,c_i\in\mathbb{F}_p$, $0\leq i\leq n-1$. Let $a=(a_0,a_1,\dots,a_{n-1})$,
$b=(b_0,b_1,\dots,b_{n-1})$ and $c=(c_0,c_1,\dots,c_{n-1})$. Then $a\in\mathcal{C}_1$, $b\in\mathcal{C}_2$ and $c\in\mathcal{C}_3$.
Assume that $\mathcal{C}_1$ is cyclic and $\mathcal{C}_2$, $\mathcal{C}_3$ are negacyclic codes. Therefore $\sigma(a)\in\mathcal{C}_1$,
$\gamma(b)\in\mathcal{C}_2$ and $\gamma(c)\in\mathcal{C}_3$. Thus $\varrho(\bar{r})=\eta_1\sigma(a)+\eta_2\gamma(b)+\eta_3\gamma(c)\in\mathcal{C}$.
Consequently $\mathcal{C}$ is a $(1-2u^2)$-constacyclic codes over $\mathcal{R}$. For the converse, let $a=(a_0,a_1,\dots,a_{n-1})\in\mathcal{C}_1$,
$b=(b_0,b_1,\dots,b_{n-1})\in\mathcal{C}_2$ and $c=(c_0,c_1,\dots,c_{n-1})\in\mathcal{C}_3$. Set $r_i=\eta_1a_i+\eta_2b_i+\eta_3c_i$, where $0\leq i\leq n-1$.
Hence $\bar{r}=(r_0,r_1,\dots,r_{n-1})\in\mathcal{C}$. Therefore $\varrho(\bar{r})=\eta_1\sigma(a)+\eta_2\gamma(b)+\eta_3\gamma(c)\in\mathcal{C}$
which shows that $\sigma(a)\in\mathcal{C}_1$, $\gamma(b)\in\mathcal{C}_2$ and $\gamma(c)\in\mathcal{C}_3$. So $\mathcal{C}_1$ is cyclic and 
$\mathcal{C}_2$, $\mathcal{C}_3$ are negacyclic codes.
\end{proof}

\begin{theorem}\label{T6}
Let $\mathcal{C}=\eta_1\mathcal{C}_1\oplus\eta_2\mathcal{C}_2\oplus\eta_3\mathcal{C}_3$ be a $(1-2u^2)$-constacyclic code of length $n$ over $\mathcal{R}$
such that $g_1(x),g_2(x),g_3(x)$ are the monic generator polynomials of $\mathcal{C}_1,\mathcal{C}_2,\mathcal{C}_2$, respectively. 
Then $\mathcal{C}=\langle\eta_1g_1(x),\eta_2g_2(x),\eta_3g_3(x)\rangle$ and $|\mathcal{C}|=p^{3n-\sum_{i=1}^3{\rm deg}(g_i)}$.
\end{theorem}
\begin{proof}
By Theorem \ref{T5}, $\mathcal{C}_1=\langle g_1(x)\rangle\subseteq\mathbb{F}_p[x]/\langle x^n-1\rangle,\mathcal{C}_2=\langle g_2(x)\rangle\subseteq\mathbb{F}_p[x]/\langle x^n+1\rangle$ and $\mathcal{C}_3=\langle g_3(x)\rangle\subseteq\mathbb{F}_p[x]/\langle x^n+1\rangle$.
Since $\mathcal{C}=\eta_1\mathcal{C}_1\oplus\eta_2\mathcal{C}_2\oplus\eta_3\mathcal{C}_3$, then 
$$\mathcal{C}=\{c(x)|c(x)=\eta_1f_1(x)+\eta_2f_2(x)+\eta_3f_3(x),~ f_1(x)\in\mathcal{C}_1,f_2(x)\in\mathcal{C}_2 \mbox{ and } f_3(x)\in\mathcal{C}_3\}.$$
Hence
$$\mathcal{C}\subseteq\langle\eta_1g_1(x),\eta_2g_2(x),\eta_3g_3(x)\rangle\subseteq\mathcal{R}_n=\mathcal{R}[x]/\langle x^n-(1-2u^2)\rangle.$$
Suppose that $\eta_1g_1(x)h_1(x)+\eta_2g_2(x)h_2(x)+\eta_3g_3(x)h_3(x)\in\langle\eta_1g_1(x),\eta_2g_2(x),\eta_3g_3(x)\rangle$, where $h_1(x),h_2(x),h_3(x)\in\mathcal{R}_n$. There exist $q_1(x)\in\mathbb{F}_p[x]/\langle x^n-1\rangle,q_2(x)\in\mathbb{F}_p[x]/\langle x^n+1\rangle$ and $q_3(x)\in\mathbb{F}_p[x]/\langle x^n+1\rangle$ such that $\eta_1h_1(x)=\eta_1q_1(x)$, $\eta_2h_2(x)=\eta_2q_2(x)$ and $\eta_3h_3(x)=\eta_3q_3(x)$.
Therefore $\langle\eta_1g_1(x),\eta_2g_2(x),\eta_3g_3(x)\rangle\subseteq\mathcal{C}$. Cosequently $\mathcal{C}=\langle\eta_1g_1(x),\eta_2g_2(x),\eta_3g_3(x)\rangle$. On the other hand $|\mathcal{C}|=|\mathcal{C}_1|\cdot|\mathcal{C}_2|\cdot|\mathcal{C}_3|=p^{3n-\sum_{i=1}^3{\rm deg}(g_i)}$.
\end{proof}

\begin{theorem}
Let $\mathcal{C}$ be a $(1-2u^2)$-constacyclic code of length $n$ over $\mathcal{R}$.
Then there exists a unique polynomial $g(x)$ such that $\mathcal{C}=\langle g(x)\rangle$ 
where $g(x)=\eta_1g_1(x)+\eta_2g_2(x)+\eta_3g_3(x)$.
\end{theorem}
\begin{proof}
Suppose that $g_1(x),g_2(x)$ and $g_3(x)$ are the monic generator polynomials of $\mathcal{C}_1,\mathcal{C}_2$ and $\mathcal{C}_3$, respectively.
By Theorem \ref{T6}, we have $\mathcal{C}=\langle\eta_1g_1(x),\eta_2g_2(x),\eta_3g_3(x)\rangle$.
Let $g(x)=\eta_1g_1(x)+\eta_2g_2(x)+\eta_3g_3(x)$. Clearly, $\langle g(x)\rangle\subseteq\mathcal{C}$. On the other hand $\eta_1g_1(x)=\eta_1g (x),\eta_2g_2(x)=\eta_2g (x)$ and $\eta_3g_3(x) =\eta_3g(x)$, whence $\mathcal{C}\subseteq\langle g(x)\rangle$. Thus $\mathcal{C}=\langle g (x)\rangle$. 
The uniqueness of $g(x)$ is followed by that of $g_1(x),g_2(x)$ and $g_3(x)$.
\end{proof}

\begin{lemma}\label{L9}
Let $x^n-(1-2u^2)=g(x)h(x)$ in $\mathcal{R}[x]$ and let $\mathcal{C}$ be the $(1-2u^2)$-constacyclic code generated
by $g(x)$. If $f(x)$ is relatively prime with $h(x)$ then $\mathcal{C}=\langle g(x)f(x)\rangle$.
\end{lemma}
\begin{proof}
The proof is similar to that of \cite[Lemma 2]{AKY}.
\end{proof}

\begin{theorem}\label{T7}
Let $\mathcal{C}=\eta_1\mathcal{C}_1\oplus\eta_2\mathcal{C}_2\oplus\eta_3\mathcal{C}_3$ be a $(1-2u^2)$-constacyclic code of length $n$ over $\mathcal{R}$
such that $g_1(x),g_2(x),g_3(x)$ are the monic generator polynomials of $\mathcal{C}_1,\mathcal{C}_2,\mathcal{C}_2$, respectively. 
Suppose that $g_1(x)h_1(x)=x^n-1$ and $g_2(x)h_2(x)=g_3(x)h_3(x)=x^n+1$ and set $g(x)=\eta_1g_1(x)+\eta_2g_2(x)+\eta_3g_3(x)$,$h(x)=\eta_1h_1(x)+\eta_2h_2(x)+\eta_3h_3(x)$. Then
\begin{enumerate}
\item $g(x)h(x)=x^n-(1-2u^2)$.
\item If $GCD(f_i(x),h_i(x))=1$ for $1\leq i\leq3$, then
$GCD(f(x),h(x))=1$ and $g(x)=g(x)f(x)$ where $f(x)=\eta_1f_1(x)+\eta_2f_2(x)+\eta_3f_3(x)$.
\end{enumerate}
\end{theorem}
\begin{proof}
(1) By assumptions we have that
\begin{eqnarray*}
g(x)h(x)&=&g(x)\big(\eta_1h_1(x)+\eta_2h_2(x)+\eta_3h_3(x)\big)\\
&=&\eta_1g_1(x)h_1(x)+\eta_2g_2(x)h_2(x)+\eta_3g_3(x)h_3(x)\\
&=&\eta_1(x^n-1)+\eta_2(x^n+1)+\eta_3(x^n+1)\\
&=&(\eta_1+\eta_2+\eta_3)x^n-(\eta_1-\eta_2-\eta_3)\\
&=&x^n-(1-2u^2).
\end{eqnarray*}
Hence, $g(x)h(x)=x^n-(1-2u^2)$.\\
(2) Suppose that $GCD(f_i(x),h_i(x))=1$ for $1\leq i\leq3$ and let $f(x)=\eta_1f_1(x)+\eta_2f_2(x)+\eta_3f_3(x)$. Then for every $1\leq i\leq 3$ there exist $a_i(x),b_i(x)\in\mathcal{R}[x]$ such that $a_i(x)f_i(x)+b_i(x)h_i(x)=1$. Set $a(x):=\eta_1a_1(x)+\eta_2a_2(x)+\eta_3a_3(x)$ and $b(x):=\eta_1b_1(x)+\eta_2b_2(x)+\eta_3b_3(x)$. Notice that $\eta_1+\eta_2+\eta_3=1$,\ $\eta_i^2=1$ and $\eta_i\eta_j=0$
for every $1\leq i\neq j\leq3$.
Thus
\begin{eqnarray*}
a(x)f(x)+b(x)h(x)&=&\eta_1[a_1(x)f_1(x)+b_1(x)h_1(x)]+\eta_2[a_2(x)f_2(x)+b_2(x)h_2(x)]\\
&+&\eta_3[a_3(x)f_3(x)+b_3(x)h_3(x)]=\eta_1+\eta_2+\eta_3=1.
\end{eqnarray*}  
It follows that $GCD(f(x),h(x))=1$. Now, by part (1) and Lemma \ref{L9}, $\mathcal{C}=\langle g(x)f(x)\rangle$. So, the uniqueness of $g(x)$ implies that
$g(x)=g(x)f(x)$.
\end{proof}

Similar to \cite[Theorem 3]{G}, we have the following theorem.
\begin{theorem}
Let $\mathcal{C}$ be a $(1-2u^2)$-constacyclic code of length $n$ over $\mathcal{R}$. Then
$$\mathcal{C}^\bot=\eta_1\mathcal{C}_1^\bot\oplus\eta_2\mathcal{C}_2^\bot\oplus\eta_3\mathcal{C}_3^\bot.$$
\end{theorem}

As a consequence of the previous theorems and \cite[Theorem 3.3]{JLU} we have the next result.
\begin{corollary}
Let $\mathcal{C}=\langle\eta_1g_1(x),\eta_2g_2(x),\eta_3g_3(x)\rangle$ be a $(1-2u^2)$-constacyclic code of length $n$ over $\mathcal{R}$ and
$g_1(x),g_2(x),g_3(x)$ be the monic generator polynomials of $\mathcal{C}_1,\mathcal{C}_2,\mathcal{C}_3$, respectively.
Suppose that $g_1(x)h_1(x)=x^n-1$ and $g_2(x)h_2(x)=g_3(x)h_3(x)=x^n+1$ and let $h(x)=\eta_1h_1(x)+\eta_2h_2(x)+\eta_3h_3(x)$. The following conditions hold:
\begin{enumerate}
\item $\mathcal{C}^{\bot}=\langle\eta_1h_1^{\bot}(x),\eta_2h_2^{\bot}(x),\eta_3h^{\bot}_3(x)\rangle$ 
and $|\mathcal{C}^{\bot}|=p^{\sum_{i=1}^3{\rm deg}(g_i)}$.
\item $\mathcal{C}^{\bot}=\langle h^{\bot}(x)\rangle$, $h^{\bot}(x)=\eta_1h_1^{\bot}(x)+\eta_2h_2^{\bot}(x)+\eta_3h^{\bot}_3(x)$,
\end{enumerate}
where for $1\leq i\leq3$, $h_i^{\bot}(x)$ is the reciprocal polynomial of $h_i(x)$, and $h^{\bot}(x)$ is the reciprocal polynomial of $h(x)$.
\end{corollary}

\begin{theorem}
Let $\mu:\mathcal{R}[x]/\langle x^n-1\rangle\to\mathcal{R}[x]/\langle x^n-(1-2u^2)\rangle$ be defined as 
$$\mu\big(c(x)\big)=c\big((1-2u^2)x\big).$$
If $n$ is odd, then $\mu$ is a ring isomorphism.
\end{theorem}
\begin{proof}
Suppose that $a(x)\equiv b(x)$ (mod $x^n-1$). Then there exists $h(x)\in \mathcal{R}[x]$ such that $a(x)-b(x)=(x^n-1)h(x)$.
Therefore 
\begin{eqnarray*}
a\big((1-2u^2)x\big)-b\big((1-2u^2)x\big)&=&\big((1-2u^2)^nx^n-1\big)h\big((1-2u^2)x\big)\\
&=&\big((1-2u^2)x^n-(1-2u^2)^2\big)h\big((1-2u^2)x\big)\\
&=&(1-2u^2)\big(x^n-(1-2u^2)\big)h\big((1-2u^2)x\big),
\end{eqnarray*}
which means if $a(x)\equiv b(x)$ (mod $x^n-1)$, then $a\big((1-2u^2)x\big)\equiv b\big((1-2u^2)x\big)$ $\big($mod $x^n-(1-2u^2)\big)$.
Now, assume that $a\big((1-2u^2)x\big)\equiv b\big((1-2u^2)x\big)$ $\big($mod
$x^n-(1-2u^2)\big)$. Then there exists $q(x)\in \mathcal{R}[x]$ such that $$a\big((1-2u^2)x\big)-b\big((1-2u^2)x\big)=\big(x^n-(1-2u^2)\big)q(x).$$
Hence
\begin{eqnarray*}
a(x)-b(x)&=&a\big((1-2u^2)^2x\big)-b\big((1-2u^2)^2x\big)\\
&=&\big((1-2u^2)^nx^n-(1-2u^2)\big)q\big((1-2u^2)x\big)\\
&=&\big((1-2u^2)x^n-(1-2u^2)\big)q\big((1-2u^2)x\big)\\
&=&(1-2u^2)(x^n-1)q\big((1-2u^2)x\big),
\end{eqnarray*}
which means if $a\big((1-2u^2)x\big)\equiv b\big((1-2u^2)x\big)$ $\big($mod $x^n-(1-2u^2)\big)$, then $a(x)\equiv b(x)$ (mod $x^n-1)$.
Consequently $a(x)\equiv b(x)$ (mod $x^n-1)\Leftrightarrow a\big((1-2u^2)x\big)\equiv b\big((1-2u^2)x\big)$ $\big($mod
$x^n-(1-2u^2)\big)$. Note that one side of the implication tells us that $\mu$ is well defined and the other side tells
us that it is injective, but since the rings are finite this proves that $\mu$ is an isomorphism.
\end{proof}

\begin{corollary}
Let $n$ be an odd natural number. Then $I$ is an ideal of $\mathcal{R}[x]/\langle x^n-1\rangle$ if and only if $\mu(I)$ is an ideal of $\mathcal{R}[x]/\langle x^n-(1-2u^2)\rangle$.
\end{corollary}

\begin{corollary}\label{C13}
Let $\mu$ be the permutation of $\mathcal{R}^n$ with $n$ odd such that 
$$\bar{\mu}(c_0,c_1,\dots,c_{n-1})=(c_0,(1-2u^2)c_1,(1-2u^2)^2c_2,\dots,(1-2u^2)^ic_i,\dots,(1-2u^2)^{n-1}c_{n-1}),$$
and $\mathcal{D}$ be a subset of $\mathcal{R}^n$. Then $\mathcal{D}$ is a cyclic code if and only if
$\bar{\mu}(\mathcal{D})$ is a  $(1-2u^2)$-constacyclic code.
\end{corollary}

\begin{definition}
Let $\tau$ be the following permutation of $\{0,1,\dots,2n-1\}$ with $n$ odd:
$$\tau=(1,n+1)(3,n+3)\cdots(2i+1,n+2i+1)\cdots(n-2,2n-2).$$
The Nechaev permutation is the permutation $\pi$ of $\mathbb{F}_p^{2n}$ defined by
$$\pi(c_0,c_1,\dots,c_{2n-1})=(c_{\tau(0)},c_{\tau(1)},\dots,c_{\tau(2n-1)}).$$
\end{definition}

\begin{proposition}\label{P15}
Let $\mu$ be defined as above. If $\pi$ is the Nechaev permutation and $n$ is odd, then $\Phi\bar{\mu}=\pi\Phi$.
\end{proposition}
\begin{proof}
Let $\bar{r}=(r_0,r_1,\dots,r_i,\dots,r_{n-1})\in\mathcal{R}^n$ where $r_i=a_i+b_iu+c_iu^2$, $0\leq i\leq n-1$. From 
$$\bar{\mu}(\bar{r})=(r_0,(1-2u^2)r_1,\dots,(1-2u^2)^ir_i,\dots,(1-2u^2)^{n-1}r_{n-1})$$ 
it follows that
$$(\Phi\bar{\mu})(\bar{r})=(-c_0,2a_1+c_1,-c_{2},2a_3+c_3,\dots,2a_{n-2}+c_{n-2},-c_{n-1},$$
$$\hspace{4cm}2a_0+c_0,-c_1,2a_2+c_2,-c_3,\dots,-c_{n-2},2a_{n-1}+c_{n-1}),$$
is equal to $(\pi\Phi)(\bar{r})$.
\end{proof}

\begin{corollary}
Let $\pi$ be the Nechaev permutation and $n$ be odd. If $\Gamma$ is the Gray image of a cyclic code over
$\mathcal{R}$, then $\pi(\Gamma)$ is a  cyclic code.
\end{corollary}
\begin{proof}
Let $\Gamma$ be such that $\Gamma=\Phi(\mathcal{D})$ where $\mathcal{D}$ is a cyclic code over $\mathcal{R}$. From Proposition \ref{P15},
$(\Phi\bar{\mu})(\mathcal{D})=(\pi\Phi)(\mathcal{D})=\pi(\Gamma)$. We know from Corollary \ref{C13} that $\bar{\mu}(\mathcal{D})$ is a 
$(1-2u^2)$-constacyclic code. Thus $(\Phi\bar{\mu})(\mathcal{D})=\pi(\Gamma)$ is a cyclic code, by Theorem \ref{T3}.
\end{proof}

Recall that two codes $\mathcal{C}_1$ and $\mathcal{C}_2$ of length $n$ over $\mathcal{R}$ are said to be equivalent if there exists a permutation $w$ of
$\{0,1,\dots,n-1\}$ such that $\mathcal{C}_2=\bar{w}(\mathcal{C}_1)$ where $\bar{w}$ is the permutation of $\mathcal{R}^n$ such that
$\bar{w}(c_0,c_1,\dots,c_i,\dots,c_{n-1})=(c_{w(0)},c_{w(1)},\dots,c_{w(i)},\dots,c_{w(n-1)})$.
\begin{corollary}
The Gray image of a cyclic code over $\mathcal{R}$ of odd length is equivalent to a cyclic code.
\end{corollary}

\begin{example}
Let $n=7$ and $x^7-1=(x-1)(x^3+x+1)(x^3+x^2+1)$ in $\mathcal{R}[x]$.
Applying the ring isomorphism $\mu$, we have $$x^7-(1-2u^2)=(x-(1-2u^2))(x^3+x+(1-2u^2))(x^3+(1-2u^2)x^2+(1-2u^2)).$$
Let $f_1=x-(1-2u^2)$ and $f_2=x^3+x+(1-2u^2)$.
If $\mathcal{C}=(f_1f_2)$, then by Theorem \ref{T3}, we know that the Gray image of the $(1-2u^2)$-constacyclic code $\mathcal{C}$ is a
cyclic code.
\end{example}

\vspace{5mm} \noindent \footnotesize 
\begin{minipage}[b]{10cm}
Hojjat Mostafanasab \\
Department of Mathematics and Applications, \\ 
University of Mohaghegh Ardabili, \\ 
P. O. Box 179, Ardabil, Iran. \\
Email: h.mostafanasab@gmail.com, \hspace{1mm} h.mostafanasab@uma.ac.ir
\end{minipage}\\

\vspace{5mm} \noindent \footnotesize 
\begin{minipage}[b]{10cm}
Negin Karimi \\
Department of Mathematics and Applications, \\ 
University of Mohaghegh Ardabili, \\ 
P. O. Box 179, Ardabil, Iran. \\
Email: neginkarimi8834@gmail.com, \hspace{1mm} neginkarimi@uma.ac.ir
\end{minipage}\\

\end{document}